\documentclass[sn-basic]{sn-jnl}

\usepackage{enumitem}
\usepackage{amsthm}
\usepackage{amsmath}
\usepackage{amssymb}
\usepackage{natbib}
\usepackage{graphicx}
\usepackage{multirow}
\usepackage{multicol}
\usepackage{comment}
\usepackage{float}

\theoremstyle{thmstyleone}%
\newtheorem{thm}{Theorem}[section]
\theoremstyle{plain}

\theoremstyle{thmstyletwo}%

\theoremstyle{thmstylethree}%
\numberwithin{equation}{section}

\begin{document}

\title[ ]{Post-model-selection prediction for GLM's}


\author[1]{\fnm{Dean} \sur{Dustin}}\email{ddustin8@huskers.unl.edu}

\author*[2]{\fnm{Bertrand} \sur{Clarke}}\email{bclarke3@unl.edu}

\affil[1]{\orgname{Charles Schwab Financial}, \orgaddress{\street{Street}, \city{Denver}, \postcode{100190}, \state{CO}, \country{USA}}}

\affil[2]{\orgdiv{Department of Statistics}, \orgname{U. Nebraska-Lincoln}, \orgaddress{\street{340 Hardin Hall North}, \city{Lincoln}, \postcode{68583-0963}, \state{NE}, \country{USA}}}

\abstract{We give two prediction intervals (PI) for Generalized Linear Models that take
model selection uncertainty into account.  The first is a straightforward extension of
asymptotic normality results and the second includes an extra optimization that
improves nominal coverage for small-to--moderate samples.  Both PI's are wider than
would be obtained without incorporating model selection uncertyainty.
We compare these two PI's with three other PI's.  Two are based on bootstrapping
procedures and the third is based on a PI from Bayes model averaging.
We argue that for general usage either the asymptotic normality or optimized
asymptotic normality PI's work best.
In an Appendix we extend our results to Generalized Linear Mixed Models.
}



\maketitle

\section{Introduction}

It is well known that linear models and their extensions - generalized linear, linear mixed,  generalized 
linear mixed models, and generlized mixed models - are the workhorses of statistical analysis. 
Aside from formulating such models, analysts have to chose amongst competing models usually
in the same class.  
Unless a model is proposed pre-experimentally, model selection from a model list is done after the 
data is collected. There are numerous model selection procedures, but regardless of which 
model is chosen as ``best'', the resulting model will have an associated variability inherited from 
the variability in the data.  How to take this variability into account properly when making predictions is the 
main topic of this paper. 

Common practice in many predictive contexts is to choose a model and then use it to generate
predictions.    Such plug-in methods are common in many modeling contexts.
This is pragmatic but neglects taking account of the uncertainty dure to model
selection or, pehaps more commonly, variable selection.  Here we propose 
prediction intervals (PI's) for generalized linear models (GLM's)
that are modified by the model selection principle (MSP) used for variable selection
so that their nomial coverage is asymptotically correct in the limit of sample sizes.
This is important because  \cite{Hong:etal:2018})
showed that using model selection procedure procedures like Akaike's information criterion
can result in predictive intervals with lower than nominal coverage if the PI's do not take
the uncertainty of the MSP into account.  

The post model selection inference problem has gained wide interest since the problem was 
first addressed in \cite{Berk:etal:2013}.  The so called post selection inference (PoSI) intervals 
introduced in \cite{Berk:etal:2013} are universally valid for any model selection principle (MSP). 
However, PoSI intervals are known to be conservative (see \cite{Leeb:etal:2015}) partially because 
they allow for any ad-hoc MSP to be used. The PoSI framework was used to construct universally 
valid (over all MSP's) confidence intervals for the mean of a predictive distribution in 
LM's in \cite{Bachoc:etal:2019}. Universally valid confidence regions for the simultaneous inference 
problem are constructed  \cite{Kuchibotla:etal:2020}.  A different approach was proposed by 
\cite{Efron:2014} that uses bootstrap intervals to address the post model selection inference 
problem under a single choice of MSP. 

\cite{stine:1985} introduced bootstrapped predictive intervals in linear regression in 1985, 
but these intervals did not consider uncertainty due to model selection.   
\cite{Leeb:2009} introduced a model selection procedure based on cross validation techniques 
and proved that using this technique, the resulting prediction interval from the selected model is 
approximately valid. While this is a seemingly strong and useful result, it holds only in the 
high sparsity case with $n << p$ as well as in the limit of large $n$. 
Specifically in his Proposition 4.3 the intervals are guaranteed to be within 
$1/\sqrt{n} + \epsilon$ of the nominal coverage for $0<\epsilon \leq \log(2)$.   More recently, 
predictive intervals based on Shorth -- i.e. the shortest interval containing a pre-specified number of 
values -- for GLM's and GAMs are studied in \cite{olive:2021}. 
While these intervals are valid,  and account for uncertainty due to the MSP,  they are not as intuitive 
and general as the ones we present in Subsec. \ref{asymptotic_PI}. 

Our methodology is in contrast to the PoSI-based intervals from \cite{Berk:etal:2013}  
that essentially widen PI's  until the nomial coverage is achieved.   Indeed, the PoSI intervals
take the pessimistic (if practical) view that model developers will use MSP's that are not 
theoretically sound.  Our approach is optimistic in that we assume a proper MSP with well known
theoretical properties will be used.  This allows us to incorporate the variability from a 
given generic MSP into our PI's.

Here,  we present two PI's that account for the uncertainty of an MSP in an intuitive manner
for GLM's. 
These PI's are easy to understand and, importantly, are easy to implement.  We also present a less 
intuitive way to construct a PI to give better coverage along the lines of PoSI intervals. 
This PI seems  work well in terms of coverage, but interpretation the interval is difficult.  

The structure of this paper is as follows.  In Sec. \ref{notation} we define the notation and setting
needed for our approach.  In fact, the notation incorporates much of the
intuition behind our approach.  In Sec. \ref{asymptotic_PI}, we present the main theorem that gives
a PI's dependent on an MSP for the case of GLM's.   We then give a finite sample improvement
for use with this PI in small-to-moderate sample settings.    We also define three other
intervals, two based on bootstrapping and one from the Bayes model average.
In Sec. \ref{simulations_glm} we present our simulation results. 
Finally, in Sec. \ref{sec_discussion} we summarize the
implications of our work.   We extend our theory to GLMM's in the Appendices.

\section{Notation and Setting}
\label{notation}

Throughout this paper we assume model selection and variable selection are synonymous, and defined as
followings.  Let $\mathcal{D}_n= \{(y_1, x_1), \ldots , (y_n,x_n) \}$ where $Y_i=y_i$ is an outcome
of the response variable and $x_i$ is a value of the $d$-dimensional explanatory variable.  We use superscripts to
indicate vectors, thus $y^n = (y_1, \ldots , y_n)^T$.
Let $m \in \mathcal{M}$ be a 
candidate model in the full collection of models $\mathcal{M}$. We define a variable selection 
procedure $M=M( \mathcal{D}_n)$ which takes the available data and maps it to a subset of 
variables based on some objective function we denote $Q$. We denote a chosen model 
$\hat{m} =\arg\min_m Q(m,\mathcal{D}_n)$.  We think of $Q$ as an objective function
such as the Akaike of Bayes information criterion (AIC, BIC)
or as a penalized loss function.   For instance, for linear models,
if $Q$ is the AIC, we have 
 \begin{equation}\label{AIC}
 Q_{AIC} (m, \mathcal{D}_n)=  -2\ln(p(y^n \vert  X^n_{m}, \hat{\beta}^{MLE}_{m})) + 2 d ,
\end{equation} 
where $d$ is the number of parameters that need to be 
estimated in model $m \in \mathcal{M}$. This defines a function
$$
M : \mathcal{D}_n \mapsto  \mathcal{M}.
$$ 
from  the data $\mathcal{D}_n$ into the model space.  That is,  we think of the
variable selection procedure as a function $M:\mathbb{R}^n \times\mathbb{R}^{n\times d}  \mapsto \mathcal{M}$. 

Other choices for $Q$ include the BIC given by
\begin{equation}
\label{BIC}
Q_{BIC} (m, \mathcal{D}_n)=  2\ln(p(y^n \vert  X^n_{m}, \hat{\beta}^{MLE}_{m})) + d \log (n),
\end{equation}
and the general Bethel-Shumway class of information criteria defined in \cite{bethel:etal:1988}.   
The BIC and the Bethel-Shumway class of information criteria are consistent for model/variable selection.

\subsection{Variable Selection}

In the predictive context, the interpretation of a parameter in a linear model is consistent across
models:  If the parameter, say, $\beta_j$, appears in multiple models it always means the change in $Y$
for a unit change in $x_j$ holding other explanatory variables constant.  This is in contrast to the 
modeling view,  see \cite{Berk:etal:2013},
that regards each parameter as an element in a whole model. and affected by each other's values.
Thus,  for us, if $\beta_j = 0$, this is mathematically equivalent to a 
model that does not include $x_j$.   Indeed, in practice, we often set a threshold $\eta >0$ and
and say that when $\vert \beta_j\vert  < \eta$,  we set $\beta_j =0$.  

With this in mind, we can define $M$ as follows. 
Let  $X = (x_1, \ldots, x_d)$, and define $M = \{\hat{\delta}_1, \ldots, \hat{\delta}_d\}$ where 
for $j= 1, \ldots , d$
$$
\hat{\delta}_j =
\begin{cases}
1 \text{ if } X_j \text{ is selected under} ~ $Q$\\
0 \text{ otherwise. }
\end{cases}
$$
For the true model we have $m_T= \{\delta_{1,T}, \ldots, \delta_{p, T}\}$ where 
$$\delta_{j, T} =
\begin{cases}
1 \text{ if } X_j \in m_T  \\
0 \text{ otherwise. }
\end{cases}
$$
Define the set $\mathcal{M}_S$ to be the set containing all possible combinations $M$ can take.  
The cardinality of $\mathcal{M}_S$ is $2^d$ and assuming $m_T$ exists,  
$m_T \in \mathcal{M}_S$.  We write
$$
\beta_T = (\beta_{\delta_{1,T}}, \ldots , \beta_{\delta_{p,T}}) \text{ and } \hat{\beta}_{M} = (\hat{\beta}_{\hat{\delta}_{1}}, \ldots , \hat{\beta}_{\hat{\delta}_{p}}),
$$ 
and  $\dim(\hat{\beta}_{M}) = \dim(\beta_{m_T})$. 
Note that if $\hat{\delta}_j = 0$, then by default we set $\hat{\beta}_{\hat{\delta_j}} =0$. 
Furthermore, we have that $\delta_{j,T} =0$  is equivalent to $\beta_j=0$. 

We specify our target of inference as $\beta_T$ and write
$$
\beta_{m_T} = \left(X'_{m_T}X_{m_T} \right)^{-1}X'_{m_T} E(Y)
$$ 
in the linear models context.  That is we are trying to estimate the true parameters, 
regardless of the model that is chosen. 
This is in contrast to the random target of inference 
$$
\beta_M = \left(X'_{M}X_{M} \right)^{-1}X'_{M} E(Y)
$$ 
defined in \cite{Berk:etal:2013}. 
Thus, for linear models,  we define the estimate for $\beta_j$ as follows: 
$$
\hat{\beta}_{\hat{M},j} = 
\begin{cases}
\left[\left(X'_{\hat{M}}X_{\hat{M}} \right)^{-1}X'_{\hat{M}} y \right]_j \text{ if } \hat{\delta}_j = 1   \\
0 \hspace{1.35in } \text{ if } \hat{\delta}_j = 1  .
\end{cases}
$$

Now there are two steps in the process of obtaining the true model. The first step is to estimate the $\delta_j$'s. In this step we want $M$ to give $\hat{\delta}_j = 1$ if $\delta_{j,T} = 1$, however, $M$ may also give  $\hat{\delta}_j = 1$ even if  $\delta_{j,T} = 0$. In this case, our definition allows the estimate $\hat{\beta}_{\delta_{j}}$ to be zero. Thus, even if $M$ includes variables that are not in $M_T$ we can still estimate their coefficients to be zero which allows $M \rightarrow m_T$ 
asymptotically (as seen in Theorem \ref{thm_AN_interval}).

\subsection{Prediction in Generalized Linear Models}
\label{sec_GLM}

As noted, we restrict attention to GLM's and GLMM's.  To be more precise,
suppose $Y \sim \mathcal{G}(\mu, R)$ where $\mathcal{G}$ is an exponential family with mean 
$\mu$ and variance $R$. Then the pdf of $Y$ given the canonical 
parameter $\theta$ is
\begin{equation}
\label{alt_exp_family}
f(y\vert \theta) = e^{\frac{y\theta - b(\theta)}{a(\phi)} + c(y, \phi)}
\end{equation}
where $\phi$ is a scale parameter.  In one parameter exponential families 
such as Poisson or Binomial distributions, $a(\phi) = 1$. 
From (\ref{alt_exp_family}) we have the following properties: 
\begin{itemize}
\item $E(Y \vert  X) = \frac{\partial b(\theta)}{\partial \theta} = \mu $ 
\item $Var(Y \vert  X) = a(\phi) \frac{\partial^2 b(\theta)}{\partial \theta^2} = a(\phi)V(\mu)$
\item $I(\theta) = Var(\ell(\theta \vert  y, \phi))$
\end{itemize}

Following standard GLM practice, we model the mean of $Y$ by transforming it to a 
linear function of the explanatory variables. The function we use to transform $E(Y)=\mu$ is 
called the link function and we denote it by $g(\cdot)$. Note that $g(\cdot)$ is a 
continuous invertible function. This gives us  the linear predictor
\begin{equation}
\eta = g(E(Y \vert  X )) =  g(\mu) = X \beta  
\end{equation}
and we define the inverse link function to be the inverse of $g(\cdot)$ which is
\begin{equation}
\mu= E(Y \vert  X ) = g^{-1}(X \beta).
\end{equation} 
For nwo, we assume that $X$ is of full rank, to avoid problems with estimability. 

Note that the canonical parameter $\theta$ is a function of $\mu$ so we write $\theta=\theta(\mu) = \theta(g^{-1}(X \beta))$. Now we can write the log-likelihood of \eqref{alt_exp_family} as
\begin{equation}
\label{log_exp_GLM}
\ell(\beta \vert  y, \phi ) = \frac{y\left(\theta(g^{-1}(X \beta))\right) - b(\theta(g^{-1}(X \beta)))}{a(\phi)} + c(y, \phi).
\end{equation}
Typically maximum likelihood along with the Newton-Raphson algorithm or Fisher scoring is used to estimate $\beta$. Usually the dispersion parameter is also unknown and must be estimated using $\hat{\phi}$.

Suppose the inferential goal is predicting the next outcome $Y^{n+1}$.
The usual point predictor under a MSP $M$ is  
\begin{equation}
\label{predictor_glm}
 \hat{Y}^{n+1}_{M} = \hat{\mu}_{M} = g^{-1}(X'^{n+1}_{M}\hat{\beta}_{M}). 
 \end{equation}
Henceforth, our focus in on constructing valid PI's for this point predictor.

\section{Candidate PI's}
 \label{asymptotic_PI}

In this section we define four PI's.  The first is derived in our Theorem \ref{thm_AN_interval}.
The second is an improvement on this interval by incorporating an extra optimization
to ensure more rapid convergence to the nominal coverage.  Both of these are in
Subsec. \ref{Main}.  In Subsec. \ref{bootstrap_procedure} we give our
third and fourth intervals are based on bootstrapping approach.
We will argue that our optimized interval provides the best performance.

\subsection{Main Result and Two PI's}
\label{Main}

One choice for a PI uses asymptotic normality of
the point predictor (\ref{predictor_glm}).    Define the statistic
\begin{equation}
\label{Z_pred_GLM}
Z_{pred} = Z_{pred}(M) = \frac{\hat{Y}^{n+1}_{M}- Y^{n+1}}{\sqrt{Var({\hat{Y}^{n+1}_{M}- Y^{n+1})}}}.
\end{equation}
We have the following result giving our first PI.

\begin{thm}
\label{thm_AN_interval}
Suppose $Y^{n}, Y^{n+1}$ come from an exponential family distribution and let $M$ be a consistent
MSP.
An asymptotically normal prediction interval for a new outcome derived from a GLM, is $PI(M)$ given by 
\begin{equation}\label{GLM_PI_AN}
 g^{-1}(X'^{n+1}_{M}\hat{\beta}_{M}) \pm z_{1-\alpha/2} \sqrt{\left[ \frac{d}{d \eta} g^{-1}(\hat{\eta}^{n+1}_{M})\right]^2 \ X'^{n+1}_{M} Var(\hat{\beta}_M)X^{n+1}_{M}+ a(\hat{\phi})_MV(\hat{\mu})_M} .
\end{equation}
\end{thm}

\begin{proof}



Asymptotically, for any fixed $m$,  

$$
\hat{\beta}_{m} \sim N\left(\beta_{m}, (X'_{m}WX_{m})^{-1}\right) 
$$
where $W =(DVD)^{-1}$.  In this notation, $V = \hbox{diag}[Var(y_i)]$ is the $n \times n $ variance matrix of 
observations, $D = \hbox{diag}[\frac{\partial \eta_i}{\partial \mu_i}]$  is the $n \times n$ matrix of derivatives 
and $\mu$ is the $n \times 1$ mean vector.  This implies 
\begin{equation}
\label{kuch_converge}
\sqrt{n}(\hat{\beta}_{m} - \beta_{m}) \overset{D}{\rightarrow} N(0, (X'_{m_T}WX_{m_T})^{-1}) 
\end{equation}
where $0 \in\mathbb{R}^{p}$ and $(X'_{m}WX_{m})^{-1}  \in \mathbb{R}^{\vert M\vert \times \vert M\vert}$.   
 
 
 

While this is useful, it is only a step toward the convergence 
$\hat{\beta}_{M} \rightarrow \beta_{m_T}$.   Since $M$ is consistent,
$m_T \in \mathcal{M}$ and hence 
 $\hat{\delta_j} \rightarrow \delta_{j, T}$ with probability 1 for all $j$ which implies $M \rightarrow m_T$. Hence, with this assumption we get an analog to (\ref{kuch_converge}) 

\begin{equation}\label{beta_converge_MT}
\sqrt{n}(\hat{\beta}_{M} - \beta_{m_T}) \overset{D}{\rightarrow} N(0, V^*_{m_T} )
\end{equation}
 where $V^*_{m_T} = (X'_{m_T}WX_{m_T})^{-1}$. Now we define the set 
 
 $$S_n = \{ \omega \vert \forall j,  \hat{\delta}_{j}(\omega) = \delta_{j, T}\} $$
and let $\textbf{1}_{S_n}$ be the indicator that $\omega \in S_n$. Further, let  $\textbf{1}_{S^c_n}$ be the indicator that $\omega$ is in the complement of $S_n$ and write 

\scalebox{.9}{\parbox{1 \linewidth}{
\begin{equation}
\label{Sn_convergence}
\sqrt{n}(X'^{n+1}_{M}\hat{\beta}_{M} - X'^{n+1}_{m_T}\beta_{m_T}) = \sqrt{n}(X'^{n+1}_{M}\hat{\beta}_{M} - X'^{n+1}_{m_T}\beta_{m_T})\textbf{1}_{S_n} + \sqrt{n}(X'^{n+1}_{M}\hat{\beta}_{M} - X'^{n+1}_{m_T}\beta_{m_T})\textbf{1}_{S^c_n}.
\end{equation}
}}

First, note that the first term on the RHS of (\ref{Sn_convergence}) becomes 
$$\sqrt{n}(X'^{n+1}_{M}\hat{\beta}_{M} - X'^{n+1}_{m_T}\beta_{m_T})\textbf{1}_{S_n} = \sqrt{n}(X'^{n+1}_{m_T}\hat{\beta}_{m_T} - X'^{n+1}_{m_T}\beta_{m_T})\textbf{1}_{S_n} $$ under consistent model selection. This term clearly converges in distribution to a normal. Namely 
$$ \sqrt{n}(X'^{n+1}_{m_T}\hat{\beta}_{m_T} - X'^{n+1}_{m_T}\beta_{m_T})\textbf{1}_{S_n} \overset{D}{\rightarrow} N\left(0, X'^{n+1}_{m_T}V^*_{m_T}X^{n+1}_{m_T}\right) $$
Now observe the second term on the RHS of (\ref{Sn_convergence}) can be written  as
\begin{align}\label{T2_Sn}
&\nonumber \sqrt{n}(X'^{n+1}_{M}\hat{\beta}_{M} - X'^{n+1}_{m_T}\beta_{m_T})\textbf{1}_{S^c_n}  \\
\nonumber& = \sqrt{n}(X'^{n+1}_{M}\hat{\beta}_{M} -X'^{n+1}_{M}\beta_{m_T} +X'^{n+1}_{M}\beta_{m_T} - X'^{n+1}_{m_T}\beta_{m_T})\textbf{1}_{S^c_n} \\
\nonumber & = \sqrt{n}(X'^{n+1}_{M}\hat{\beta}_{M} -X'^{n+1}_{M}\beta_{m_T})\textbf{1}_{S^c_n} +\sqrt{n}(X'^{n+1}_{M}\beta_{m_T} - X'^{n+1}_{m_T}\beta_{m_T})\textbf{1}_{S^c_n}\\
 & = X'^{n+1}_{M}\sqrt{n}(\hat{\beta}_{M} -\beta_{m_T})\textbf{1}_{S^c_n} +\sqrt{n}(X'^{n+1}_{M} - X'^{n+1}_{m_T})\beta_{m_T}\textbf{1}_{S^c_n}.
\end{align}
Using (\ref{beta_converge_MT}) and the fact that  $X'^{n+1}_{M}$ is bounded, we know the first in (\ref{T2_Sn}) converges in distribution to a normal.  Also, we see that in the second term in (\ref{T2_Sn}), $\beta_{m_T}$ is a bounded constant vector,  $(X'^{n+1}_{M} - X'^{n+1}_{m_T}) $ is bounded and $P(S^c_n) \overset{p}{\rightarrow} 0$ by assumption. Thus we see  asymptotically that (\ref{Sn_convergence}) is 
\begin{align}\label{Sn_convergence_infty}
\nonumber \sqrt{n}(X'^{n+1}_{M}\hat{\beta}_{M} - X'^{n+1}_{m_T}\beta_{m_T}) & \overset{\infty}{=} N(0,X'^{n+1}_{m_T}V^*_{m_T}X^{n+1}_{m_T})  \textbf{1}_{S_n}  \\
\nonumber &+ N(0,X'^{n+1}_{m_T}V^*_{m_T}X^{n+1}_{m_T} ) \textbf{1}_{S^c_n} + \sqrt{n}  \textbf{1}_{S^c_n} 
\end{align}
Now we need to show that $\sqrt{n}  \textbf{1}_{S^c_n} = o_p(1) $. First note that the union of events bound gives
$$
P(S^c_n) \leq \sum^p_{j=1}  P( \vert  \hat{\delta}_j - \delta_{j, T}\vert  > \eta ),
$$ for some $\eta>0$, so using symmetry in the MSP it is enough to show $$\lim_{n \rightarrow \infty } P( \vert  \hat{\delta}_j -   \delta_{j,T}  \vert > 1/\sqrt{n} ) = 0 $$ for any $j$. 
It is easy to see that 
\begin{equation}\label{littleopSc}
\lim_{n \rightarrow \infty } P( \vert  \hat{\delta}_j -   \delta_{j,T}  \vert > 1/\sqrt{n} )  = \lim_{n \rightarrow \infty } P( \vert  \hat{\delta}_j -   \delta_{j,T}  \vert = 1 )   
\end{equation}
because $\delta_{j,T}$ and $\hat{\delta}_j$ are either 1 or 0.  Now  because we have chosen a consistent MSP we have $$ \lim_{n \rightarrow \infty } P( \vert  \hat{\delta}_j -   \delta_{j,T}  \vert = 1 )   = 0. $$ Hence, the left hand side of (\ref{littleopSc}) is also equal to zero, implying that  $\sqrt{n}  \textbf{1}_{S^c_n} = o_p(1) $.
Now, Slutsky's theorem gives us, 
$$
\sqrt{n}(X'^{n+1}_{M}\hat{\beta}_{M} - X'^{n+1}_{m_T}\beta_{m_T}) \overset{D}{\rightarrow} N\left(0, X'^{n+1}_{m_T}V^*_{m_T}X^{n+1}_{m_T}\right).
$$
Now to get a predictive distribution, we observe that the delta method gives us  
 \begin{equation}\label{delta_var}
 \sqrt{n}\left(g^{-1}(X'^{n+1}_{M}\hat{\beta}_{M}) - g^{-1}(X'^{n+1}_{m_T}\beta_{m_T} ) \right) \overset{D}{\rightarrow} N \left(0,  \left[ \frac{d}{d \eta} g^{-1}(\eta^{n+1}_{m_T})\right]^2  X'^{n+1}_{m_T} V^*_{m_T}X^{n+1}_{m_T}  \right)
 \end{equation}
 where $\eta^{n+1}_{m_T} = X'^{n+1}_{m_T}\beta_{m_T}$
 Thus, we see that although the GLM estimates are biased we still get convergence in distribution when model selection occurs in the $\mathcal{M}$-closed case. 
The variance of $\hat{Y}^{n+1}_{M} - Y^{n+1} $ is 
\begin{align}\label{pred_int_var}
\nonumber Var(\hat{Y}^{n+1}_{M} - Y^{n+1}) &= Var(\hat{Y}^{n+1}_{M}) + Var(Y^{n+1})\\
\nonumber& =  Var(g^{-1}(X'^{n+1}_{M}\hat{\beta}_{M})) + a(\phi)V(\mu)\\
&= \frac{1}{n}\left[ \frac{d}{d \eta} g^{-1}(\eta^{n+1}_{m_T})\right]^2  X'^{n+1}_{m_T} V^*_{m_T}X^{n+1}_{m_T} + a(\phi)V(\mu)
\end{align}
due to (\ref{delta_var}). Again, because $Y^{n+1}$ is a random variable, and not a parameter, we must consider the variance of it as well, which we get assuming it will come from the exponential family distribution as $Y_1, \ldots, Y_n$. This quantity, however, is impossible to compute because we do not know $m_T$. Hence, we must replace $m_T$ with $M$, making the variance a random quantity that depends on model selection.  
Now we use (\ref{Z_pred_GLM}) as a pivotal quantity to get 

 \scalebox{.88}{\parbox{1\linewidth}{
 \begin{align}\label{Pred_interval_derive}
 \hspace*{-.65cm} \nonumber 1-\alpha & \leq P\left(\left \vert  Z_{pred} \right \vert  < z_{1-\alpha/2} \right)\\
 \nonumber&= P\left(\left \vert  \hat{Y}^{n+1} - Y^{n+1}  \right \vert  < z_{1-\alpha/2} \sqrt{Var(\hat{Y}^{n+1} - Y^{n+1})} \right)\\
 &= P\left( \hat{Y}^{n+1} - z_{1-\alpha/2}\sqrt{Var(\hat{Y}^{n+1} - Y^{n+1})} < Y_{n+1}  <  \hat{Y}^{n+1}  + z_{1-\alpha/2} \sqrt{Var(\hat{Y}^{n+1} - Y^{n+1})} \right) .
 \end{align} 
 }}
 Hence using (\ref{pred_int_var})
$$
 \left[ g^{-1}(X'^{n+1}_{M}\hat{\beta}_{M}) \pm z_{1-\alpha/2} \sqrt{\frac{1}{n}\left[ \frac{d}{d \eta} \widehat{g^{-1}(\eta^{n+1}_{M})}\right]^2 \ X'^{n+1}_{M} V^*_M X^{n+1}_{M}+ a(\hat{\phi}_M)V(\hat{\mu}_M)} \right ]
$$
is a $100(1-\alpha) \%$ prediction interval for $Y^{n+1}$
\end{proof}

We now offer an improvement on the PI from Theorem \ref{thm_AN_interval}.  Note that thisn result
uses the standard normal quantile to define the predictive interval, but this is not the only choice. 
Instead, we can adjust the width of the interval to correct to poor coverage. To do this, we write the interval  as 
\begin{eqnarray}
&& PI(C_{\alpha, M}) = g^{-1}(X'^{n+1}_{M}\hat{\beta}_{M}) \pm \nonumber \\
&&C_{\alpha,M} \sqrt{\left[ \frac{d}{d \eta} g^{-1}(\hat{\eta}^{n+1}_{M})\right]^2 \ X'^{n+1}_{M} Var(\hat{\beta}_M) X^{n+1}_{M}+ a(\hat{\phi})V(\hat{\mu})}
\label{GLM_PI_PoSI_AN}
\end{eqnarray}
where $C_{\alpha, M}$ is chosen to satisfy 
\begin{equation}\label{C_alpha}
C_{\alpha, M} = \arg\min_{C} P\left( Y^{n+1} \in PI(M,C)\right)
\end{equation}
for all $C$ such that $P\left( Y^{n+1} \in PI(M,C)\right) \geq 1 - \alpha$. Importantly, this 
probability also sees the random variable $M$ and hence inherits the uncertainty 
associated with $M$ as well and $Y^{n+1}$. This is in the same spirit of the 
PoSI constant in \cite{Berk:etal:2013}. That is, we enlarge $C_{\alpha,M}$ to 
account for the uncertainty in $M$.  We can approximate $C_{\alpha, M}$ using Monte Carlo cross validation.

We begin by choosing an interval on $\mathbb{R}^+$ denoted  $\mathcal{C} = [ C_{low}, C_{high} ]$ that we will perform the line search on to estimate  $C_{\alpha, M} $. Next, we randomly splitting $\mathcal{D}_n$ into $L \in \mathbb{N}$  test and train sets, $\mathcal{D}_{train, \ell}$ and $\mathcal{D}_{test, \ell}$ for $\ell = 1, \ldots, L$. Then for each $\ell$ we estimate $\beta$ by
$$
\hat{\beta}_{\ell} = (X'_{train, M,\ell}X_{train,M,\ell})^{-1}X'_{train,M,\ell}y_{train,\ell}
$$
using $\mathcal{D}_{train, m}$, and form the predictor  $\hat{Y}^{test}_{M,\ell} = g^{-1}(X_{test, M ,\ell} \hat{\beta}_{\ell} )$. Now we form the prediction interval $PI_{M,\ell}(C)$ namely $ \hat{Y}^{n+1}_{M,\ell} \pm$
 \begin{eqnarray}
 C \sqrt{\left[ \frac{d}{d \eta} g^{-1}(\hat{\eta}^{n+1}_{test,M, \ell})\right]^2 \ X'^{n+1}_{test,M,\ell} Var(\hat{\beta})_{test,M,\ell} X^{n+1}_{test,M,\ell}+ a(\hat{\phi})_{M,\ell}V(\hat{\mu})_{M,\ell}},
\label{optimized}
\end{eqnarray}
and for each $C \in \mathcal{C}$, check if $y_{test, \ell} \in PI_{M,\ell}(C)$. 
Then we choose the value $C$ that gives us $1-\alpha$ coverage for the 
Monte Carlo samples. More formally, we can approximate $C_{\alpha, M} $ by 
\begin{equation}
\label{C_opt_MC}
\hat{C}^{MC} = \arg\min_{C}  \frac{1}{L} \sum^L_{\ell=1} \left \vert\frac{1}{\#(\mathcal{D}_{test,\ell})}\sum^{\#(\mathcal{D}_{test,\ell})}_{i=1} I_{y_{test,\ell} \in PI_{M,\ell}(C)} - (1-\alpha) \right \vert
\end{equation}
where $I_{y_{test,\ell} \in PI_{M,\ell}(C)}$ is the indicator that the test values are in the constructed intervals. 

The intuition behind using this interval in place of the PI in Theorem \ref{thm_AN_interval} is that, in finite 
samples the difference between $z_{1-\alpha/2}$ and $\hat{C}^{MC}$ can be interpreted as t
he added variability due to model uncertainty.

\subsection{Two Bootstrap Based PI's}
\label{bootstrap_procedure}

In the frequentist setting, perhaps the most natural way to obtain a prediction interval that takes into 
account the uncertainty of both model selection and the uncertainty associated with the 
distribution of the new outcome is to make use of the bootstrap. 
Accordingly,  to form our fiest bootstraped PI,
we use the bootstrap to estimate the distribution of  
\begin{equation}\label{mu_M}
\hat{\mu}_M= E(Y^{n+1}\vert X^{n+1}_{M}) = g^{-1}(X^{n+1}_M\hat{\beta}_M),
\end{equation}
and $ a(\hat{\phi})_M $.  Then for each bootstrapped mean and dispersion function, we generate a 
new observation from the distribution of $Y^{n+1} \vert X^{n+1}, \mu, \phi $, i.e. $\mathcal{G}$.   
 
 Let $\hat{p}(\hat{\mu})$ denote be the bootstrapped density of  (\ref{mu_M})  and $\hat{p}(a(\hat{\phi}))$ be the bootstrapped density of $ a(\hat{\phi})_M $. Then $\hat{p}(Y^{n+1})$ be the resulting estimated density of $Y^{n+1}$. 
 
 
 
 The procedure is as follows, 
\begin{itemize}
\item obtain $B$ bootstrap replications of $\hat{\mu}_M$, denoted $\mu^*_{1}, \ldots , \mu^*_{B}$
\item obtain $B$ bootstrap replications of $a(\hat{\phi})_M$, denoted $a(\phi)^*_{1}, \ldots , a(\phi)^*_{B}$
\item generate $y^*_{1}(\mu^*_{1},a(\phi)^*_{1}), \ldots, y^*_{B}(\mu^*_{B},a(\phi)^*_{B})$, from $\mathcal{G}$.
\end{itemize}
The sample $y^*_1, \ldots, y^*_B$ can be used to estimate an approximate marginal predictive distribution for $Y^{n+1}$.  To obtain the PI, we use the appropriate percentile interval from this distribution. That is, to obtain a $100(1-\alpha) \%$ PI we use the interval 
\begin{equation}\label{PI_boot_new}
[q^*_{1-\alpha/2}, q^*_{\alpha/2}]
\end{equation}
 where $q^*_{\alpha}$ is the $\alpha$ quantile from $\hat{p}(Y^{n+1})$. The use of $\hat{p}(\hat{\mu})$ and $\hat{p}(a(\hat{\phi}))$  to obtain the  estimated predictive distribution $\hat{p}(Y^{n+1})$ allows $\hat{p}(Y^{n+1})$ to inherit the variability from $\hat{p}(\hat{\mu})$, $\hat{p}(a(\hat{\phi}))$ and the variability that is already associated with the known parametric distribution $\mathcal{G}$. Hence, the interval (\ref{PI_boot_new}) is typically widened due to the uncertainty of the model selection procedure as well as the uncertainty of the distribution of $Y^{n+1}$. 

Now, in the GLM setting, coverage for the PI in (\ref{PI_boot_new}) should be closer to 
the $1-\alpha$ nominal coverage than the PI resulting from ignoring the model 
uncertainty. 
Note that as $n\rightarrow \infty$ the variability due to model uncertainty 
will go to 0 and this interval will converge to the standard PI. 
Bootstrap PI's for the Gaussian case are studied in a fairly narrow 
(small $d$ and moderate $n$) setting in \cite{Hong:etal:2018_2}. These authors suggest 
that in this setting the bootstrap distribution fails to assess the uncertainty of model selection accurately.
 We explore different simulation settings to evaluate the performance of 
bootstrap intervals in Subsec. \ref{simulations_glm}.

Our second bootstrapped PI is formed as follows.  Recall
he interval in Theorem \ref{thm_AN_interval} is a random  
because it depends on $M$. It is directly usable for predictions, but we must use 
$\hat{M}$ in place of $M$ to get a confidence statement. 
Nevertheless, we provide an approximate interval by ``smoothing'' over $M$, which accounts for the uncertainty of $M$ in both the center and width of the interval. This is similar to the approach used in \cite{Efron:2014} for estimation. The method we propose is to use $\hat{p}(\hat{\mu})$, the bootstrap the distribution of $\hat{\mu}_M=g^{-1}(\hat{\eta}_M)$ as described 
earlier in this Subsection, to obtain an approximation for the 
predictor and its variance that accounts for model selection uncertainty. 

Specifically, we use 
$$
\tilde{\mu}  = \frac{1}{B}\sum^B_{b=1} \mu^*_b
$$
for the point predictor. We approximate the variance of $\hat{\mu}_M$ with
 \begin{align*}
 Var(\mu^*) & = \widehat{Var}(\hat{\mu}_M)   \\
& = \frac{1}{B-1} \sum^B_{b=1} \left(\mu^*_b-\tilde{\mu}   \right)^2  \\
 & \approx \left[ \frac{d}{d \eta} g^{-1}(\hat{\eta}^{n+1}_{M})\right]^2 \ X'^{n+1}_{M} Var(\hat{\beta}_M) X^{n+1}_{M},
 \end{align*}
 and the estimated variance of the predictive distribution is given by 
 \begin{align*}
 Var(Y^*) & =  \widehat{Var}(Y^{n+1}) \\
 & =   \frac{1}{B-1} \sum^B_{b=1} \left(y^*(\mu^*_{b}) - \bar{y}^*  \right)^2 \\
& \approx a(\hat{\phi}_M)V(\hat{\mu}_M)
 \end{align*}
 where $\bar{y}^*=   \frac{1}{B-1} \sum^B_{b=1} y^*(\mu^*_{b})$.  We treat $Y^*$ as a random variable approximating $Y^{n+1}$. Note also that we are required to estimate  $a(\hat{\phi}_M)V(\hat{\mu}_M)$ in (\ref{GLM_PI_AN}), but this again is a random variable so we using bootstrapping to account for the uncertainty in $M$ is necessary for this term also. 
Now as an ad-hoc fix, we rewrite (\ref{GLM_PI_AN}) to give our second bootstrapped PI
\begin{equation}
\label{GLM_PI_boot_AN}
PI(M) = \tilde{\mu} \pm z_{1-\alpha/2} \sqrt{  Var(\mu^*) + Var(Y^*))}.
\end{equation}

\section{Simulation Results for GLM's}
\label{simulations_glm}

We give two contexts in which the PI's we have defined in (\ref{GLM_PI_AN}), (\ref{GLM_PI_PoSI_AN}),
(\ref{PI_boot_new}), (\ref{GLM_PI_boot_AN}) can be readily found.
 Respectively, these intervals are labeled the asymptotic normal PI (AN),  the optimized AN $\hat{C}^{MC}$,
the bootstrapped (boot) PI, and the `smoothed' asymptotic normal (S-AN) PI.  In addition to the intervals we 
have derived, we give the BMA PI's as well as the `Naive' PI's obtained by applying the inverse link to 
a confidence interval for the mean on the linear predictor scale; this is often done by practitioners 
as a pragmatic solution. 

 In Sec. \ref{subsec_gaussian} we present these intervals for the standard Gaussian case and in Sec. \ref{subsec_binom}, we present the prediction intervals for binomial regression, i.e., a more general case of logistic regression.  
For both cases we use 500 new observations from their respective distribution and calculate 
the estimated predictive coverage using 

\begin{equation}\label{coverage_sims}
\widehat{coverage} = \frac{1}{500} \sum^{500}_{i =1} I_{y^{new}_{i} \in PI_i(X^{new}_i, X^n, y ^n)}, 
\end{equation}
where each $PI_i(X^{new}_i, X^n, y ^n)$ depends on the data and the new observed explanatory variables.  For the PIs that require bootstrapping we resample the data 500 times to obtain the bootstrapped distributions. 

\subsection{Gaussian Linear Models }
\label{subsec_gaussian}

 In the standard case, we assume $Y \sim N(\mu, \sigma^2)$, and the log likelihood is
$$
L(\mu_i, \sigma^2 \vert  y_i) =\frac{y_i\mu_i - (\mu_i^2/2)}{\sigma^2} - \left(\frac{y^2_i}{2\sigma^2} + \log(\sigma \sqrt{2\pi}\right),
$$
the canonical parameter is $\theta_i = \mu_i$, $b(\theta) = \mu+i^2/2$, $a(\phi) = \sigma^2$, and $V(\mu_i) = 1$. 
The linear predictor uses the identity link function and the point predictor is
$$
\hat{Y}^{n+1}_{M} = X'^{n+1}_{M}\hat{\beta}_M.
$$
The asymptotic normal PI from (\ref{GLM_PI_AN}) for $Y^{n+1}$ is
$$
PI(M) =  \left[\hat{Y}^{n+1}_{M} \pm z_{1-\alpha/2} \hat{\sigma}_M \sqrt{X'^{n+1}_{M}(X'_{M}X_{M})^{-1}X^{n+1}_{M} + 1} \right].
$$

Our simulation results for Gaussian data includes coverage and width estimates for the normal PI in (\ref{GLM_PI_AN}),
the PI \eqref{optimized} using $\hat{C}^{MC}$,  the bootstrap PI in (\ref{PI_boot_new}),  and the `smoothed'
normal interval (\ref{GLM_PI_boot_AN}).    We do not include the Naive interval because in the Gaussian case
it is equivalent to AN.   For the interval using $\hat{C}^{MC}$, we do a grid search for the value of $C_{\alpha, M}$ 
on the interval from 1.95 to 5 in increments of 0.05.  

The simulation setup is as follows. First, we consider two model selection procedures, BIC and AIC. Both methods are implemented in R using the step() function by setting the respective penalties for BIC and AIC. We also use BMA implemented with the  {\textsf{BAS} } package in R. We consider various choices for $n$ (30,50,100,200) and choose $p=25$. We randomly generate values for $\sigma$ and $\beta$ once, and fix those values throughout the simulations. Accordingly, let 
$$
\beta = (\beta_1, \ldots, \beta_{25})' = (6.43, 4.39,4.26,4.11, 0,\ldots, 0)'
$$
and $\sigma = 0.93$. We simulate $n$ observations for  the design matrix $X$  according to 

$$
X \sim MVN_p(0, I_p),
$$
and then draw and  $n\times 1$ vector of observations from  $Y \sim N(X\beta, \sigma^2I_n)$.   
We then calculate estimated coverage using (\ref{coverage_sims}). Ideally, we want coverage close to $0.95$.  
When choosing between competing PI's with good coverage, we prefer the one with the narrowest width.
The results are seen in Table \ref{simulations_Guassian}.

\begin{table}[h!]
\centering
\caption{{\textbf{Simulation results for Gaussian data with $p=25$ and $p_0 = 4$. }}}
\begin{tabular}[t]{|c|c|c|c|c|}
\hline
 $n$ &MSP &Interval & Coverage  & Avg.Width (SE) \\
\hline
50 &AIC & AN    &   0.88        & 3.63 (0.17)  \\
  & & S-AN   & 1  & 14.2 (3.81) \\
 & & boot &  0.99 & 8.8 (1.87) \\
     &       & $\hat{C}^{MC}$  &  0.98     & 5.5 (0.35) \\
    &BIC     & AN     & 0.88 & 3.63 (0.18) \\
    &     &S-AN     & 1  & 13.79 (3.80) \\
    &           &boot &  0.99&  8.4 (1.78)\\
     &       & $\hat{C}^{MC}$  &   0.98   & 5.37 (0.26) \\
    & BMA &  & 0.89 &  3.98 (0.11)\\
\hline
100 &AIC & AN    &   0.91        & 3.63 (0.08)  \\
 & & S-AN    & 0.95 & 4.43 (0.41) \\
 & & boot &  0.92 & 3.81 (0.26) \\
     &       & $\hat{C}^{MC}$  &   0.93   & 4.08 (0.09) \\
  &BIC &  AN    &   0.92    & 3.70 (0.06)  \\
 &  &S-AN    &  0.94 & 4.25  (0.33)\\
 & & boot & 0.92  &  3.74 (0.22) \\
     &       & $\hat{C}^{MC}$  &   0.94   & 3.97 (0.05) \\
   & BMA  & & 0.93 & 3.73 (0.05) \\
\hline
200 &AIC & AN    &   0.92    & 3.70 (0.05)  \\
 & & S-AN   & 0.93  & 3.95 (0.16) \\
 & & boot & 0.91  & 3.64 (0.14) \\
     &       & $\hat{C}^{MC}$  &   0.93   &  3.96 (0.06) \\
 &BIC &  AN    &   0.94    & 3.75 (0.03)  \\
 & & S-AN  & 0.94 &  3.91 (0.14)\\
 & & boot &  0.92 & 3.65 (0.12) \\
 &       & $\hat{C}^{MC}$  &   0.94   & 3.92 (0.03) \\
   & BMA  & & 0.92 & 3.74  (0.03) \\
\hline
\end{tabular}
\label{simulations_Guassian}
\end{table}%

\begin{table}[h!]
\centering
\caption{{\textbf{Gaussian cross validation results for the optimal choice for $\hat{C}^{MC}$.}}}
\begin{tabular}[t]{|c|c|c|c|c|}
\hline
n & MSP &$\hat{C}^{MC}$ \\
\hline
50 &  AIC & 2.95  \\
     &  BIC & 2.90  \\
\hline
100 &  AIC &  2.20 \\
       &  BIC &  2.10 \\
\hline
200 &  AIC &  2.10 \\
      &  BIC &  2.05 \\
\hline
\end{tabular}
\label{C_opt_gaussian}
\end{table}%

Note that the differences between using AIC and BIC are negligible, so we describe the performance 
of each PI only once (rather than once for each MSP). It is seen in Table \ref{simulations_Guassian} that 
AN has low coverage for $n =50$, but gets close to the nominal coverage for the larger sample 
sizes.  
For $n=50$, both S-AN and boot give at least the nominal coverage and arguably reasonable 
width of PI's to be useful. Here,  $\hat{C}^{MC}$ gives close to the stated 95\% coverage and is 
noticeably narrower than both S-AN and boot, so it is the preferred PI.

When $n=100$ and 200, we observe all of the 5 PIs are roughly equal in terms of coverage and width. Since AN is the easiest to implement as it does not require any bootstrapping or cross validation, we recommend using it with relatively large $n$. For intermediate $n$ we recommend using $\hat{C}^{MC}$ as it gives appropriate coverage and is narrower than the other PIs. 

We give the optimal choices for $\hat{C}^{MC}$ for each sample size in Table \ref{C_opt_gaussian}. We observe that as sample size increases, $\hat{C}^{MC}$ decreases as expected. This reflects the fact that as we gather more data, the uncertainty in model selection also decreases. 



\subsection{Binomial Regression}\label{subsec_binom}

Suppose we have $n$ independent but not identically distributed random variables
 following $Y_i  \sim Bin(r_i,p_i)$ so $E(Y_i) = r_i p_i$. We write $W=\frac{Y_i}{r_i}$ as our response 
to model the proportion of success, and then we convert back to number of successes to form 
our predictive interval.  Now we have $E(W) = p_i$ and the 
log likelihood for a given $i$  is given by 
$$
L(p_i \vert  w_i) = \frac{ w_i \log\left(\frac{p_i}{1-p_i}\right) + \log(1-p_i)}{\frac{1}{r_i}} + \log\binom{r_i}{nw_i},
$$
which reveals the canonical parameter $$\theta_i = logit(p_i) = \log\left(\frac{p_i}{1-p_i}\right).$$ We also see that $a(\phi) =\frac{1}{r_i}$, $b(\theta_i) = - \log(1+e^{\theta_i}) = -\log(1-p_i)$, and thus $$V(p_i) = \frac{\partial^2 b(\theta_i) }{\partial p_i^2} =  \frac{p_i(1-p_i)}{r_i}.$$ Thus the linear predictor is defined by the logit link as 
$$E\left(\frac{Y_i}{r_i}\right) = g(p_i) = \log\left(\frac{p_i}{1-p_i}\right)  = X'_i\beta
$$ 
and the inverse link function, which gives the probability of success, is given by 
$$
p_i= g^{-1}(X'_i\beta)  = \frac{1}{1+ e^{-X'_i\beta}}. 
$$
Of course, we do not know $p_i$, so we estimate $p_i$ by 
$$
\hat{p}_i = g^{-1}(X'_i\hat{\beta}) = \frac{1}{1+ e^{-X'_i\hat{\beta}}}.
$$
Given $n$ observations $Y_1, \ldots , Y_n$, our goal is to predict the total 
number of successes  $Y_{n+1}$ in $r_{n+1}$ trials while accounting for 
model selection. We denote the predicted probability of success $\hat{p}^{n+1}_{M}$ and its value is given by 
$$
\hat{p}^{n+1}_{M} = \frac{1}{1+ e^{-X'^{n+1}_{M}\hat{\beta}_{M}}}.
$$ 


Recalling that 
$$
E(Y_{n+1}) = r_{n+1} \cdot g^{-1}(X'^{n+1}\beta) = r_{n+1} \cdot p_{n+1},
$$
the form of the post-model selection AN PI for a binomial random variable is
\begin{align}
&& PI(M) =  r_{n+1} \cdot \hat{p}^{n+1}_{M} \pm  \nonumber \\
&& z_{1-\alpha/2} \cdot r_{n+1} \sqrt{ \frac{ e^{-2\hat{\eta}^{n+1}_{M}}}{\left(1+ e^{-\hat{\eta}^{n+1}_{M}}\right)^4}X'^{n+1}_{M} Var(\hat{\beta}_M)X^{n+1}_{M}+ \frac{1}{r^{n+1}}\hat{p}_{M}\left(1-\hat{p}_{M}\right) }
\label{PI_logistic}
\end{align}
where the factor $r_{n+1}$ in the width of the intervals comes from the the distribution in (\ref{delta_var}) being multiplied by this factor. The interval in (\ref{PI_logistic}) gives a prediction interval for total number of successes in $r_{n+1}$ trials. 

In the setting described above, our simulations are as follows. 
Let  $X \sim MVN_p(0,I_p)$ and 
$$
\beta = (\beta_1, \ldots, \beta_{25})' = (0.252, 0.171, -0.268, 0.09, 0, \ldots 0)'.
$$ 
Now  we calculate the estimated coverage using (\ref{coverage_sims}).  Again, we want 
coverage close to $0.95$ and narrow width. The simulated results are given in Table \ref{simulations_binomial_30}.

\begin{table}[h!]
\centering
\caption{{\textbf{Simulation results for binomial data with $r = 30$. }}}
\begin{tabular}[t]{|c|c|c|c|c|}
\hline
 $n$ &MSP & Interval & Coverage  & Avg Width (SE)  \\
\hline
50 &AIC & Naive &  .51 & 4.9  (1.35) \\
 & & AN &  0.83 & 9.51  (0.93) \\
 &  & S-AN   & 1 & 28.31 (2.22) \\
 & & boot & 1 & 20.82  (3.35)\\
 &       & $\hat{C}^{MC}$  &   0.97   & 15.55 (0.78) \\
 &BIC & Naive &  .48 &  4.03 (1.08) \\
  & & AN & 0.90  & 9.56  (0.76) \\
 &  & S-AN   &  1& 23.36 (3.93) \\
 & & boot &  1& 17.81  (3.06)\\
 &       & $\hat{C}^{MC}$  &   0.98   & 14.35 (0.74) \\
\hline
100 & AIC &  Naive &  .46 & 3.29  (0.86) \\
 &  &  AN &0.94  & 9.42 (0.87) \\
 &  & S-AN   & 0.99 &  14.11(2.22) \\
 & & boot & 0.99 &   11.99(1.66)\\
 &       & $\hat{C}^{MC}$  &   0.99   & 13.01 (0.61) \\
 &BIC & Naive &  .43 &  2.77 (0.69) \\
  & & AN & 0.95  &  9.38 (0.79) \\
 &  & S-AN   & 1 &  13.38 (2.00) \\
 & & boot & 0.99 &  11.57 (1.47)\\
 &       & $\hat{C}^{MC}$  &   0.99   & 12.67 (0.54) \\
\hline
200 &AIC & Naive &  .37  & 2.43 (0.90) \\
 & & AN &0.94   & 8.91 (1.60) \\
 &  & S-AN  &  0.97 &  10.08(2.10) \\
 & & boot & 0.97  &  9.37 (1.96) \\
 &       & $\hat{C}^{MC}$  &   0.94   & 9.15 (0.87) \\
 &BIC & Naive &  .37 & 2.18  (0.79) \\
  & & AN &  0.94 &  8.93 (1.56) \\
 &  & S-AN   &0.97  & 9.75 (1.92) \\
 & & boot & 0.98 &  9.30 (1.82)\\
 &       & $\hat{C}^{MC}$ &   0.95   & 9.01 (0.83) \\
\hline
\end{tabular}
\label{simulations_binomial_30}
\end{table}%

\begin{table}[h!]
\centering
\caption{{\textbf{Binomial cross validation results for the optimal choice for $\hat{C}^{MC}$.}}}
\begin{tabular}[t]{| c| c| c| c| c| }
\hline
n & MSP &$\hat{C}^{MC}$ \\
\hline
50 &  AIC & 5.00  \\
     &  BIC & 4.40  \\
\hline
100 &  AIC &  4.10 \\
       &  BIC &  3.90 \\
\hline
200 &  AIC &  3.00 \\
      &  BIC &  2.85 \\
\hline
\end{tabular}
\label{C_opt_binom}
\end{table}%


For $n=50$, the Naive interval has very poor coverage for both AIC and BIC. Using AIC and the AN interval results in undercoverage, but it is much better than the Naive PI. This is also true using BIC as the MSP. Both S-AN and boot are conservative, give coverage larger than the stated coverage. The width of both S-AN and boot make the intervals fairly uninformative despite having better coverage than Naive and AN. Finally, we observe $\hat{C}^{MC}$ performs noticeably better than the other PI's. This suggests that the cross validation step to widen the asymptotic normal PI is useful. 

Looking at the $n=100$ and case,  we see the Naive interval is worse than the smaller sample case. AN gives very good coverage and the smallest width among all of the PIs for both AIC and BIC.  The other 3 PIs, S-AN, boot, and $\hat{C}^{MC}$ give close stated coverage but they are slightly wider than the AN interval. 
When $n=200$ we see Naive is by far the worst among the 5 PIs, but the other 4 are roughly the same with AN and $\hat{C}^{MC}$ having perhaps slightly better coverage and narrower PIs than S-AN and boot. 

These results confirm two main points. First, the AN PI achieves the stated 95\% coverage as given in Theorem \ref{thm_AN_interval} when the sample size is large enough. Second $\hat{C}^{MC}$ always gives appropriate coverage, and appears to reduce to AN as $n$ increases. This leads us to recommend using $\hat{C}^{MC}$, especially for intermediate sample sizes, and use AN for large $n$.

Again, we list the optimal cross validation constants in Table \ref{C_opt_binom}. As in the Gaussian case, we see that $\hat{C}^{MC}$ decreases as the sample size increases. However, in the Binomial case, $\hat{C}^{MC}$  is noticeably larger than the Gaussian case. This may be due to the fact that we are using a normal PI for data that is not normal.

\section{Discussion}
\label{sec_discussion}

Our main contribution is the PI in Theorem \ref{thm_AN_interval}, and the 
small sample correction using $\hat{C}^{MC}$ given in \eqref{optimized}.
 Much of the literature on GLM prediction has focused on $\textit{confidence}$ intervals around predictors, which we refer to as the `Naive' PI, rather than true prediction intervals. That is, it is common for analysts to apply the inverse link function to the endpoints of a confidence interval in the linear predictor scale. This approach does not account for uncertainty appropriately because it uses the variability on the linear predictor scale rather than the data scale. Here we have presented prediction intervals that are derived on the model-scale rather than the linear-predictor scale. 

We have presented several prediction intervals that consider model uncertainty. The PI derived in Theorem \ref{thm_AN_interval} accounts for model uncertainty via the consistency of the MSP. This PI severely underperforms in terms of predictive coverage in small sample size, e.g. $n \approx p$, cases but as $n \rightarrow \infty$ the predictive coverage is roughy the nominal $1 -\alpha$ coverage. The boot and S-AN PIs tend to be too wide, suggesting far too much model uncertainty, for these PIs to be useful. These two PIs overcorrect the width of the intervals for the amount of the uncertainty in the model selection considered here. Again as $n$ increases, both boot and S-AN become usable 
(due to the MSP choosing the correct model). At this point, the bootstrapping is not necessary, 
however.  Since AN performs well with large samples, there is no need to bootstrap, we can directly use AN. 


Taken together, our results provide valid post model selection PIs for GLM's for moderate and large samples.


\begin{subappendices}

\subsection{Extension to GLMM's }
\label{sec_GLMM}


The approach described in Sec. \ref{sec_GLM} to obtain valid prediction intervals after model selection extends naturally to the class of generalized linear mixed models. Here we assume the random variable $Y\vert X,\beta,Z,U \sim {\mathcal{G}}$ where ${\mathcal{G}}$ is a distribution in the exponential family. 
We write the linear predictor as  
\begin{equation}\label{linear_predictor_GLMM}
\eta =  g(E(Y\vert U)) = g(\mu \vert  U) =  X\beta + ZU
\end{equation}
where $\beta$ is the vector of fixed effects, $U$ is a random effect such that $U \sim N(0, \Sigma_U)$.  $X$ and $Z$ are their respective design matrices.   
The mean function is 
$$
\mu = E(Y \vert  U )=g^{-1}(\eta) = g^{-1}(X\beta + ZU) 
$$
and the variance is 
$$
Var(Y\vert U) = V^{1/2}_{\mu} A V^{1/2}_{\mu}
$$
where $V^{1/2}_{\mu} = \hbox{diag}\left[\sqrt{V(\mu)}\right]$ and $A = \hbox{diag}\left[1/a(\phi)\right]$. 

As with GLM's, model selection is often performed when forming predictors. In the GLMM setting 
model selection can be done on both $X$ and $Z$, however here we focus on model selection 
on the design matrix $X$.
Analogous to the GLM case, we state an asymptotic normal predictive interval For GLMM's that is derived in the same way as the GLMM. The only difference is the random effects part of the linear predictor. However, recall the random effects have expectation 0, so the location of the asymptotic distribution does not change. The variance, on the other hand, does increase. This is seen in (\ref{GLMM_PI_AN} ) as the width has an extra term for the variance of the random effects:  We get that
$PI(M,C_{\alpha}) = g^{-1}(\hat{\eta}^{n+1}_M) \pm $
\begin{equation}\label{GLMM_PI_AN}
C_{\alpha} \sqrt{\left[ \frac{d}{d \eta} g^{-1}(\hat{\eta}^{n+1}_{M})\right]^2 \left(X'^{n+1}_M Var(\hat{\beta}_M)X^{n+1}_M + Z'^{n+1}Var(\hat{u})Z^{n+1}\right) + a(\hat{\phi})V(\hat{\mu})}.
\end{equation}


\subsection{GLMM bootstrap intervals}

As with the GLM AN interval, we can approximate the variance of $g^{-1}(\hat{\eta}^{n+1}_M)$ using bootstrapping and replace  (\ref{GLMM_PI_AN}) with 
\begin{equation}\label{GLMM_PI_AN_boot}
PI(M,C_{\alpha}) = g^{-1}(\hat{\eta}^{n+1}_M) \pm z_{1-\alpha/2} \sqrt{\hat{Var}(g^{-1}(\hat{\eta}^{n+1}_M))^{boot}+ a(\hat{\phi})V(\hat{\mu})}
\end{equation}
where $\hat{Var}(g^{-1}(\hat{\eta}^{n+1}_M))^{boot}$ is simply the variance of the bootstrapped distribution of $g^{-1}(\hat{\eta}^{n+1}_M)$.

We can also use the bootstrap approach, in same way as with the GLM, to obtain a bootstrap distribution for a new outcome. In the GLMM setting, we bootstrap the expected value of the distribution for a new outcome,
$$
\hat{\mu}_M = g^{-1}(X'^{n+1}_M\hat{\beta} + Z'^{n+1}\hat{u}).
$$
We proceed as follows:
\begin{itemize}
\item obtain $B$ bootstrap replications of $\hat{\mu}_M$, denoted $\mu^*_{1}, \ldots , \mu^*_{B}$,
\item obtain $B$ bootstrap replications of $a(\hat{\phi})_M$, denoted $a(\phi)^*_{1}, \ldots , a(\phi)^*_{B}$,
\item generate $y^*_{1}(\mu^*_{1},a(\phi)^*_{1}), \ldots, y^*_{B}(\mu^*_{B},a(\phi)^*_{B})$, from $\mathcal{G}$.
\end{itemize}
The sample $y^*_1, \ldots, y^*_B$ is used to obtain the predictive interval by extracting the appropriate percentile interval from this distribution. Thus the PI is
\begin{equation}\label{GLMM_PI_boot_new}
[q^*_{1-\alpha/2}, q^*_{\alpha/2}],
\end{equation}
the $1-\alpha/2$ and $\alpha/2$ quantiles from $y^*_1, \ldots, y^*_B$  which inherits the uncertainty of $M$, $\hat{\beta}$ and $\hat{u}$. 
These intervals are implementable assuming we already have estimates $\hat{\beta}$ and $\hat{u}$. 
Regardless of which method is used to form predictors, we theoretically can use 
both intervals (\ref{GLMM_PI_AN_boot}) or (\ref{GLMM_PI_boot_new}) because 
predictors and mean and variance functions, as well as the uncertainty associated with the MSP 
can be obtained through the bootstrap procedure. Thus, a closed form solution of parameter estimates is 
not necessary to obtain valid PI's. 
Also, we can, at least theoretically, still use 
both of the intervals presented to account for the uncertainty of 
model selection in the random effects design matrix.

\subsection{Computational issues for GLMM's}
\label{comp_issues_GLMM}

The theoretical and bootstrap based intervals we havew proposed
to capture the uncertainty of model selection are not implementable, at least yet. This is due to 
convergence issues with implementing GLMM's. 
Estimation in GLMM's requires integrating out the random effects, and these integrals do 
not have closed form solutions. Thus, numerical integration is necessary, making the integrals computationally hard.  

In practice, now, there is no a single best approach so one tries many approaches until the algorithm converges. 
Once convergence is achieved, classical approaches to assess model fit are used. In the 
bootstrapping approach, we require estimation over many repeated samples of the data and this 
would require convergence of the estimates in the GLMM over each resample. The estimates 
require numerical integration for each resample of the data which requires a person trying several 
algorithms until one works. We attempted this, but we were unsuccessful because 
convergence in each resample using a fixed numerical integration method is not feasible.


\end{subappendices}

\backmatter


\bmhead{Acknowledgments}

The first author acknowledges funding from the University of Nebraska Program of Excellence in Computational Science.

\bibliography{references}


\end{document}